\documentclass[a4paper,11pt]{article}
\usepackage{amsmath,amsthm,amsfonts,amssymb, color}
\usepackage{graphicx}
\usepackage[latin1]{inputenc}
\usepackage[all]{xy}
\usepackage{enumerate}
\usepackage{graphicx}
\usepackage{natded}

\usepackage{bbold}

\usepackage{amsthm}
\usepackage{float}
\usepackage[bottom]{footmisc}
\usepackage{authblk}

\newtheorem{theorem}{Theorem}[section]
\newtheorem{lem}[theorem]{Lemma}
\newtheorem{cor}[theorem]{Corollary}
\newtheorem{prop}[theorem]{Proposition}
\newtheorem{thm}[theorem]{Theorem}

\theoremstyle{definition}
\newtheorem{defn}[theorem]{Definition}

\newtheorem{rem}[theorem]{Remark}

\newcommand{\ra}{\rightarrow}

\newcommand{\aS}{\mathsf{ASA}}
\newcommand{\clos}{\mathrm{\bf{Clos}}}

\DeclareMathAlphabet\mathbfcal{OMS}{cmsy}{b}{n}

\newcommand{\dpa}{\mathsf{DP}}

\newcommand{\hd}{\mathsf{3DP}}
\newcommand{\hdz}{\mathsf{C3DP}}
\newcommand{\ctl}{\mathsf{C3\L}}

\newcommand{\cmv}{\mathsf{C3MV}}

\newcommand{\cb}{\bf{C(B)}}
\newcommand{\cbs}{C(B)}

\newcommand{\bool}{\mathsf{BA}}

\newcommand{\vbtsm}{\mathsf{B3SM}}

\title{Assume-guarantee contract algebras are dp-algebras}

\author[1, 2]{\small Jos\'e Luis Castiglioni%\footnote{Corresponding author: e-mail: jlc@mate.unlp.edu.ar; tel.: +54 221 4245875}
}
\author[3]{\small Rodolfo C. Ertola-Biraben}

\affil[1]{\footnotesize Departamento de Matem\'atica, Facultad de Ciencias Exactas, Universidad Nacional de La Plata, Argentina}
\affil[2]{\footnotesize CONICET-Argentina}
\affil[3]{\footnotesize CLE, Universidade Estadual de Campinas (UNICAMP), Campinas, SP, Brazil}

\date{}

\begin{document}

\maketitle

\begin{abstract}
In [Incer Romeo, I. X., \textit{The Algebra of Contracts}. Ph.D. Thesis, UC Berkeley (2022)] an algebraic perspective on assume-guarantee contracts is proposed. 
This proposal relies on a construction involving Boolean algebras. 
However, the structures thus proposed lack a clearly prescribed set of basic operations, necessary if we want to see them as a class of algebras (in the sense of Universal Algebra). 
In this article, by prescribing a suitable set of basic operations on contracts, we manage to describe these algebras as (a generating set of members of) well-known varieties.  
\end{abstract}

\section{Introduction}
\label{intro}
The aim of this note is to observe that the class of assume-guarantee contract algebras, as presented in \cite{IR}, can be seen as a particular class of dp-algebras (as appearing in \cite{TK}). 
More concretely, we show that choosing an adequate signature for assume-guarantee contract algebras, they may be regarded as dp-algebras. 
Furthermore, any centred three-valued dp-algebra (see Section \ref{DPA} for its definition) is isomorphic to an assume-guarantee contract algebra as defined in Section \ref{AGC}. 
Hence, from an algebraic point of view, we can abstractly identify  assume-guarantee contract algebras and centred three-valued dp-algebras.

The notion of contract derives from the theory of abstract data types and was first suggested by Meyer \cite{BM} in the context of the programming language Eiffel, 
following the original ideas introduced by Floyd and Hoare \cite{CH} to assign logical meaning to sequential imperative programs in the form of triples of assertions. 
%A Hoare triple $\{A\}c\{B\}$ consists of properties $A$ and $B$ and a program $c$. 
%The program is correct if $B$ holds after the execution of $c$, assuming that $A$ holds before this execution. 

\vskip.3cm
In this paper we will follow \cite{IR} for introducing the basic ideas regarding assume-guarantee contracts. 
We refer the reader to that work, \cite{BCNP} and \cite{IBV} for further details. 

As in this paper we will only consider assume-guarantee contracts, from now on they will be called contracts.
This will be the subject of Section \ref{AGC}.

The definition and properties of the classes of dp-algebras of interest for this work are presented in Section \ref{DPA}. 

In Section \ref{OSM-G}, we show that the variety of c3dp-algebras ($\hdz$) is categorically equivalent to that of Boolean algebras. 
Since it may be seen (see Section \ref{G3D}) that there is a term-equivalence between the varieties of c3dp-algebras and bounded three-valued Sugihara monoids, the mentioned fact may also be obtained from \cite{GR12}, where the authors prove that the variety of bounded odd Sugihara monoids is categorically equivalent to that of G\"odel algebras. 
However, in order to be self-contained and also to avoid overly complex and powerful mathematical machinery for the purpose of this paper, we provide a direct and simple proof.
The mentioned equivalence allows us to show that every c3dp-algebra is isomorphic to an contract algebra. 

In Section \ref{G3D}, we show that the variety $\hdz$ is term equivalent to other well-studied varieties and, in consequence, contract algebras may also be seen (up to term equivalence) as algebras in these other varieties.

Finally, in Section \ref{FS} we compare our work with the adjunction in \cite{IRp}.

\section{Contracts} \label{AGC}

Contracts provide a formal framework for contract-based design (see \cite{BM}). 
Contracts assume that all the functional and non-functional behaviours of the system have been modelled in advance as elements of an underlying set of behaviours. 
Two concepts to implement compositional design with contracts are required:
\begin{itemize}
\item a notion of order on contracts and
\item an operation to obtain system specifications from the specifications of subsystems.
\end{itemize}

Following this intuition, an algebraic description of contracts is proposed in \cite{IR} (see also \cite{IBV}). 
A contract algebra would be some very particular type of lattice ordered monoid, where the lattice structure renders account of the notion of order on contracts and the monoid operation is that allowing the obtention of system specifications from the specifications of subsystems. 
However, in order to simplify computations, we will consider an equivalent formulation in a different signature.

\begin{defn} \label{contract}
Let $\mathbf{B} = (B; \cap, \cup,\ ', 0, 1)$ be a Boolean algebra. 
An \emph{(assume-guarantee) contract} $C$ on $\mathbf{B}$ is a pair $(a, g) \in B \times B$.
We call $a$ the assumption of $C$ and $g$ the guarantee of $C$. 
\end{defn}

Our definition of contract is based on Chapter 6 of \cite{IR}. 

Let us define the following relation taken from \cite[p. 147]{BCNP} and \cite[p. 19]{IR}. 

\begin{defn}
Given two contracts $C_1 = (a_1, g_1)$ and $C_2 = (a_2,g_2)$, we say that $C_1$ \emph{refines} $C_2$ if $a_2 \leq a_1$ and $g_1 \vee \neg a_1 \leq g_2 \vee \neg a_2$.
\end{defn}

It is clear that the refinement relation is reflexive and transitive, that is, a preorder relation.
We will use the notation $C_1 \cong C_2$ to mean that the contracts $C_1$ and $C_2$ refine each other.

The following notion appears in \cite[p. 185]{BCNP} and \cite[p. 19]{IR}.

\begin{defn}
Given a contract $C = (a, g)$ we will say that $C$ is \emph{saturated} if $g = g \vee \neg a$.
\end{defn}

\begin{lem}
Given a contract $C = (a, g)$ the following conditions are equivalent.

(i) $C$ is saturated,

(ii) $a = a \vee \neg g$,

(iii) $g = a \to g$,

(iv) $a \vee g = 1$,

(v) $\neg a \wedge \neg g = 0$.
\end{lem}

It is easily seen that given a contract $C = (a, g)$, the contract $C^*= (a, g \vee \neg a)$ is saturated and $C^* \cong C$.

From now on, contracts stands for saturated contracts.

Let us write $\cbs$ for the set of (saturated) contracts on the Boolean algebra $\mathbf{B}$. 
We have that $(\cbs, \preceq)$ is an order.
Indeed, $(a_1,g_1) \preceq (a_2,g_2)$ iff $a_2 \leq a_1$ and $g_1 \leq g_2$.
Let us refer to this poset, simply as $\cb$. 

It is easy to see that the poset $\cb$ is in fact a bounded distributive lattice with respect to the order defined. 
Indeed, the elements $(1,0)$ and $(0,1)$ behave as bottom and top, respectively.
And for $(a_1,b_1)$ and $(a_2, b_2)$ in $\cbs$, the pairs $(a_1 \cup a_2, b_1 \cap b_2)$ and $(a_1 \cap a_2, b_1 \cup b_2)$ behave as the infimum and the supremum, respectively.
We shall also write $\cb$ for the corresponding bounded distributive lattice. 

It can also be seen that both the meet- and the join-pseudo-complement exist and are given by $\neg (a,b) = (b, b')$ and $D(a,b) = (a', a)$, respectively.
 
Moreover, the Stone and co-Stone equalities and the regularity inequality hold, as we have $\neg (a,b) \vee \neg \neg (a,b) = (0,1)$, $D(a,b) \wedge DD(a,b) = (1,0)$, and $(a,b) \cap D(a,b) \preceq (c,d) \vee \neg (c,d)$.

Finally, $(1,1)$ is the only element $(a,b)$ satisfying $\neg (a,b) = \bot$ and $D(a,b) = \top$.

The following definition is essentially implicit in \cite{IR}. 
Our contribution is to choose a suitable language for the algebras involved. 

\begin{defn} \label{calg}
Let $\mathbf{B} = (B; \cap, \cup,\ ', 0, 1)$ be a Boolean algebra. 
We define the \emph{contract algebra} on $\mathbf{B}$ as the algebra ${\bf{C(B)}} = (C(B); \wedge, \vee, \neg, D, \bot, \gamma, \top)$ of type $(2, 2, 1, 1, 0, 0, 0)$ 
where $C(B) = \{ (a,b): a,b \in B $ and $a \vee b = 1\}$ and the operations are defined as follows:

$(a,b) \wedge (c,d) = (a \cup c, b \cap d)$, 

$(a,b) \vee (c,d) = (a \cap c, b \cup d)$, 

$\neg (a,b) = (b, b')$,

$D(a,b) = (a', a)$,

$\bot = (1,0)$,

$\gamma = (1,1)$,

$\top = (0,1)$.
\end{defn}

\begin{rem}
In \cite{IR} the author introduces so-called AGC-algebras without fixing any signature. 
We have fixed a signature and called them contract algebras.
\end{rem}

In Section \ref{DPA} we shall equationally characterize contract algebras. 
We shall introduce  the class of c3dp-algebras, which may be seen to be the variety generated by the class of contract algebras. 
In fact, we shall see that any c3dp-algebra is isomorphic to a contract algebra. 
This allows us to see c3dp-algebras as a sort of abstract contract algebras.

 \begin{rem}
It is well noted in \cite{IR} that different semiring structures can be defined on $C(B)$ turning it into an involutive residuated lattice.
Particularly, a unary operation $(\ )^{-1}$ (called reciprocal in \cite[p. 72]{IR}) and a binary operation $\|$ (called composition in \cite[p. 73]{IR}) turn $(C(B); \wedge, \vee, \|, (\ )^{-1}, (1,1))$ into an odd Sugihara monoid (see \cite{GR12} and \cite{MRW} for reference). 
However, these operations are definable in terms of $\neg$ and $D$ as

$(a,b)^{-1} := (b,a) = D(a,b) \wedge [(a,b) \vee \neg (a,b)]$ (see \cite[p. 72]{IR}) and

$(a_1, g_1) \| (a_2, g_2) := \ ((a_1 \cap a_2) \cup (g_1 \cap g_2)',g_1 \cap g_2) = ((a_1, g_1) \wedge \neg \neg (a_2, g_2)) \vee (\neg \neg (a_1, g_1) \wedge (a_2, g_2))$ (see \cite[p. 73]{IR}).

\noindent In fact, $(C(B); \wedge, \vee, \|, (\ )^{-1}, (1,1))$ becomes a bounded three-valued Sugihara monoid (see Definition \ref{tBOSM}).
Thus, the machinery developed in \cite{GR12} and further worked in \cite{GR15} and \cite{FG} allows to prove that the variety in the signature $\{\wedge, \vee, \|, (\ )^{-1}, \bot, \gamma, \top\}$ is categorically equivalent to Boole and the functor witnessing this equivalence is (on objects) $C:B \mapsto C(B)$. 
Three-valued Sugihara monoids are term equivalent to c3dp-algebras (presented in Section \ref{DPA}) and, hence, here we opted, for the sake of the reader, to provide a direct and elementary proof of the equivalence between the categories of Boole and c3dp-algebras, which implies the aforementioned one between three-valued Sugihara monoids and Boole (see Section \ref{G3D}).
\end{rem}

\begin{rem} \label{RO}
In \cite{DC} the author introduced the term orthopair for a notion previously studied by many other authors in several environments. 
Given a Boolean algebra {\bf{B}}, an orthopair is a pair $(a,b) \in B \times B$ such that $a \wedge b = 0$.
Furthermore, in \cite[Section 3.2]{DC} the mentioned author endows the set $O(B)$ of orthopairs  on a Boolean algebra with the structure of BZ lattices. 
It is easy to note that these BZ lattices can be endowed with a new constant $c$ satisfying $-c = c$. 
This expansion is term equivalent to the structure of c3dp-algebras (see next Section).
Moreover, an easy computation shows that for any Boolean algebra ${\bf{B}}$, the map $C(B) \to O(B)$ given by $(a,b) \mapsto (a',b')$ is an isomorphism of c3dp-algebras.

%The order in ${\bf{O(B)}}$ is defined by $(a_1,b_1) \leq (a_2, b_2)$ iff $a_1 \leq a_2$ and $b_2 \leq b_1$.

In \cite{ABCG} the authors note that orthopairs may be endowed with an IUML-algebra structure, which is known to be term equivalent to bounded odd Sugihara monoids. 
\end{rem}

\section{dp-algebras} \label{DPA}

\begin{defn}
A \emph{double p-algebra} (dp-algebra for short) is an algebra $(A; \wedge, \vee, \neg, D, 0, 1)$ of type (2, 2, 1, 1, 0, 0)  such that
	\begin{itemize}
		\item[]$(A; \wedge, \vee, 0, 1)$ is a bounded distributive lattice,
		\item[]for every $a \in A$, $\neg a$ is the greatest element $b$ of $A$ such that $a \wedge b = 0$, 
		\item[]for every $a \in A$, $D a$ is the smallest element $b$ of $A$ such that $a \vee b = 1$.
	\end{itemize}
\end{defn}

It is well known that dp-algebras form a variety, which we denote by $\dpa$.

\begin{rem} 
Let us briefly state that dp-algebras where already studied by Skolem in \cite[\S 2]{TS}.
\end{rem}

An $(A; \wedge, \vee, \neg, D, 0, 1) \in \dpa$ is said to be \emph{regular} if for every $a, b \in A$, $a \wedge Da \leq b \vee \neg b$, where $\leq$ stands for the usual lattice order. 
Clearly, regular dp-algebras form a subvariety of $\dpa$. 
As a brief historical remark, let us note that the notion of regularity was introduced by Malcev in \cite{AM} and regular dp-algebras were studied by Varlet in \cite{V1972} and Katrinak in \cite{TK}. 

Let $(A; \wedge, \vee, \neg, D, 0, 1)$ be a dp-algebra. 
An element $a \in A$ is said to be \emph{central} if $\neg a = 0$ and $Da = 1$. 
A dp-algebra with its universe having a central element is said to be a \emph{centred dp-algebra} (cdp-algebra for short).
It is easily seen that regularity and centrality are independent concepts.
However, in a regular dp-algebra there is at most one central element. 

A \emph{three-valued dp-algebra} (3dp-algebra for short) is a regular dp-algebra that satisfies the equations
\begin{itemize}
	\item[(S)] $\neg x \vee \neg\neg x = 1$ (Stonean equation), 
	\item[(cS)] $D x \wedge D D x = 0$ (coStonean equation).
\end{itemize}

3dp-algebras form a subvariety of regular dp-algebras, which we denote by $\hd$. 

Three-valued dp-algebras were studied by Varlet in \cite{V1972} and Katrinak in \cite{TK}.

The variety of centred three-valued dp-algebras, c3dp-algebras for short, which is denoted by $\hdz$, is an expansion of $\hd$ with the $0$-ary operation $c$ satisfying the equations $\neg c = 0$ and $Dc = 1$. 
It is shown in \cite{CE} that $\hd$ is the variety generated by the totally ordered dp-algebra with three elements ($\bot < m < \top$). 
Similar arguments show that $\hdz$ is the variety generated by the totally ordered centred dp-algebra with three elements. 
Furthermore, c3dp-algebras are functionally complete (straightforwardly from results in \cite{CE}).

A straightforward computation shows the following fact.

\begin{prop}
Given a Boolean algebra ${\bf{B}}$, the algebra ${\bf{C(B)}}$ is a c3dp-algebra.
\end{prop}

\section{A categorical equivalence} \label{OSM-G}

In this section we prove in a direct and elementary way that there exists a categorical equivalence between the categories of Boolean and c3dp-algebras.
This equivalence is witnessed by the functor ${\bf{C}}: {\bf{\bool}} \to {\bf{\hdz}}$ that assigns to each Boolean algebra ${\bf{B}}$  the c3dp-algebra ${\bf{C(B)}}$ as in Definition \ref{calg} and to any Boolean morphism $f:B_1 \to B_2$ the map $C(f): C(B_1) \to C(B_2)$ given by $C(f) (a,b) = (fa,fb)$.

Let us now define the functor ${\bf{(\ )^-}}: {\bf{\hdz}} \to {\bf{\bool}}$ that assigns to each c3dp-algebra ${\bf{A}} = (A; \wedge, \vee, \neg, D, 0, \gamma, 1)$ the Boolean algebra ${\bf{A^-}} = (A^-; \wedge, \vee, ', 0, 1)$ with $A^- : = \{x \in A:  x \leq \gamma \}$, $\wedge$ and $\vee$ are the restrictions of those of ${\bf{A}}$, $x' := x \wedge \gamma$ and $1 := \gamma$. 
Moreover,  given a c3dp-morphism $g: A_1 \to A_2$ the map $g^-: A_1^- \to A_2^-$ is its restriction $g/A_1^-$.

Let us now see that the given functors form an equivalence pair, that is, we have to prove that ${\bf{A}} \cong {\bf{C(A^-)}}$ and ${\bf{[C(B)]^-}} \cong {\bf{B}}$ for every algebra in the corresponding class.

\begin{prop} \label{iso1}
Let ${\bf{A}} = (A; \wedge, \vee, \neg, D, 0, \gamma, 1)$ be a c3dp-algebra. 
It holds that ${\bf{A}} \cong {\bf{C(A^-)}}$.
\end{prop}

\begin{proof}
Let us define $\varphi(a) = (Da \wedge \gamma, a \wedge \gamma)$, for $a \in A$.

The map $\varphi$ is well defined as both $Da \wedge \gamma$, $a \wedge \gamma \in A^-$ and $(Da \wedge \gamma) \vee (a \wedge \gamma) = \gamma$.

Let us see that $\varphi$ is a homomorphism. 

\

We have that $\varphi(a \wedge b) = (D(a \wedge b) \wedge \gamma, a \wedge b \wedge \gamma) = 
((Da \vee Db) \wedge \gamma, a \wedge b \wedge \gamma) =  
((Da \wedge \gamma) \vee (Db \wedge \gamma), a \wedge \gamma \wedge b \wedge \gamma) =  
(Da \wedge \gamma, a \wedge \gamma) \ \wedge (Db \wedge \gamma, b \wedge \gamma) = 
\varphi(a) \wedge  \varphi(b)$. 

\

Similarly, we have that $\varphi(a \vee b) = \varphi(a) \vee  \varphi(b)$. 

\

$\varphi(\neg a) = (D\neg a \wedge \gamma, \neg a \wedge \gamma)$ and 
$\neg \varphi(a) = \neg (Da \wedge \gamma, a \wedge \gamma) = (a \wedge \gamma, [a \wedge \gamma]') $. 
So, it is enough to see that $D\neg a \wedge \gamma = a \wedge \gamma$ and that $\neg a \wedge \gamma = [a \wedge \gamma]'$.
Now, on the one hand, $D\neg a \wedge \gamma = \neg \neg a \wedge \gamma = a \wedge \gamma$. 
On the other hand, $[a \wedge \gamma]' = \neg (a \wedge \gamma) \wedge \gamma = (\neg a \vee \neg \gamma) \wedge \gamma = \neg a \wedge \gamma$.

\

$\varphi(Da) = (DDa \wedge \gamma, Da \wedge \gamma)$. 
Now, $DDa = \neg Da = \neg (Da \wedge \gamma)$ and so, $\varphi(Da) = (\neg (Da \wedge \gamma) \wedge \gamma, Da \wedge \gamma)$ whence $\varphi(Da) = ((Da \wedge \gamma)', Da \wedge \gamma) = D\varphi(a)$.
 
 \
 
 $\varphi(0) = (\gamma,0)$.

\

$\varphi(\gamma) = (\gamma, \gamma)$.

\

$\varphi(1) = (0,\gamma)$.

\

For inyectivity, suppose  $\varphi(a) = \varphi(b)$, that is, $(Da \wedge \gamma, a \wedge \gamma) = (Db \wedge \gamma, b \wedge \gamma)$. 
Then, $Da \wedge \gamma = Db \wedge \gamma$ and $a \wedge \gamma = b \wedge \gamma$ whence  $\neg Da = \neg Db$ and $\neg a = \neg b$, respectively.
Now, $\neg Da = \neg Db$ implies that $Da = Db$. 
So, by regularity it follows that $a = b$, as desired.

\

For suryectivity, let $(a,b) \in C(A^-)$  whence $a \vee b = \gamma$ which implies that 
$Da = Db = 1$ as $D\gamma = 1$ and also implies by absorption that $a \wedge \gamma = a$ and $b \wedge \gamma = b$.
Let us now see that $\varphi (\neg a \vee b) = (a,b)$. 
It is enough to see that $D(\neg a \vee b) \wedge \gamma = a$ and $(\neg a \vee b) \wedge \gamma = b$.
On the one hand, $D(\neg a \vee b) \wedge \gamma = (D\neg a \wedge Db) \wedge \gamma = D\neg a \wedge \gamma = \neg \neg a \wedge \gamma = a \wedge \gamma = a$. 
On the other hand, by distributivity  $(\neg a \vee b) \wedge \gamma = (\neg a \wedge \gamma) \vee (b \wedge \gamma) = (\neg a \wedge \gamma) \vee b = [\neg a \wedge (a \vee b)] \vee b$. 
Now,  as $\neg a \wedge (a \vee b) \leq b$, it follows that  $[\neg a \wedge (a \vee b)] \vee b= b$.
\end{proof}

\begin{cor} \label{cor1}
Every c3dp-algebra is isomorphic to a contract algebra.
\end{cor}

\begin{cor} \label{cor2}
The variety $\hdz$ is generated by the class of contract algebras.
\end{cor}

\begin{prop} \label{iso2}
Let ${\bf{B}}$ be a Boolean algebra. 
It holds that ${\bf{B}} \cong {\bf{[C(B)]^-}}$.
\end{prop}

\begin{proof}
It is immediate to check that, taking $x \in B$, the map $\psi(x) = (1,x)$ is well defined and is an isomorphism.
\end{proof}

Using Propositions \ref{iso1} and \ref{iso2} we get the following result.

\begin{thm} \label{cateq}
The categories $\bool$ and $\hdz$ are equivalent.
\end{thm}

\begin{rem}
As was stated in the Introduction, the categorical equivalence proved in this section could be obtained adapting results from  \cite{GR12}. 
It is worth noting (see next section) that the categories of c3dp-algebras and bounded three-valued Sugihara monoids are equivalent.
Moreover, there is a categorical equivalence between the categories of bounded odd Sugihara monoids and G\"odel algebras as proved in \cite{GR12}.
Restricting this equivalence to one between Boolean algebras and bounded three-valued Sugihara monoids, it is possible to elicit the desired categorical equivalence.
Particularly, (see Theorem 6.4. in \cite{GR12}), any bounded three-valued Sugihara monoid is isomorphic (as bounded Sugihara monoid) to a contract algebra, and hence, $\vbtsm$ is the variety generated by the class of contract algebras. 
\end{rem}

\

In the next section we present alternative characterizations for contract algebras, possibly, with a different signature.

\section{Some varieties term-equivalent to $\hdz$} \label{G3D}

In this section different well-known varieties of algebras of relevance in this area is presented. 
We recall some basic properties of the members of these varieties. 
Furthermore, since these varieties are shown to be term-equivalent to $\hdz$, they turn out to be alternative abstract descriptions for contract algebras.

\subsection{Odd Sugihara monoids}

Odd Sugihara monoids are well-known involutive commutative residuated lattices, of interest in the study of certain relevance logics \cite{MRW}. 
Their lattice reducts are De Morgan algebras. 

The following definitions are based in \cite{MRW}. 
Note that the same classes of algebras appear in \cite{GR12} using a different signature.

\begin{defn} \label{OSM}
An \emph{odd Sugihara monoid} is an algebra $(A; \wedge, \vee, \cdot, \sim, e)$ of type (2,2,2,1,0) such that 
\begin{itemize}
		\item[\textbf{A1}] $(A; \vee, \wedge)$ is a distributive lattice,
		\item[\textbf{A2}] $(A; \cdot, e)$ is a commutative monoid such that $x \leq x \cdot x$,
		\item[\textbf{A3}] $\sim \sim x = x$,
		\item[\textbf{A4}] if $z \cdot x \leq y$, then $z\ \cdot \sim y \leq\ \sim x$,
		\item[\textbf{A5}] $\sim e = e$.
		%\item[\textbf{A6}] $(x\ \cdot \sim y) \wedge (y\ \cdot \sim x)\leq e$.
	\end{itemize}
\end{defn}

\begin{defn} \label{BOSM}
A \emph{bounded odd Sugihara monoid} is an algebra $(A; \wedge, \vee, \cdot, \sim, 0, e, 1)$ of type (2,2,2,1,0,0,0) such that 
\begin{itemize}
		\item[\textbf{}] $(A; \vee, \wedge, 0, 1)$ is a bounded lattice,
		\item[\textbf{}] $(A; \wedge, \vee, \cdot, \sim, e)$ is an odd Sugihara monoid.
	\end{itemize}
\end{defn}

It is well known that the class of bounded odd Sugihara monoids forms a variety.

\begin{prop} \label{basicfacts}
For any $a, b$ in a bounded odd Sugihara monoid, 
\begin{enumerate}
	\item[(i)] $a \wedge b \leq a \cdot b$,
	\item[(ii)] $a\ \cdot \sim a \leq e$,
	\item[(iii)] $a\ \wedge \sim(a \cdot 1) = 0$,
    \item[(iv)] $(a\ \cdot \sim b) \wedge (b\ \cdot \sim a)\leq e$.
\end{enumerate}
\end{prop}

\begin{proof}
Due to \textbf{A2}, we have that $a \wedge b \leq (a \wedge b) \cdot (a \wedge b)$. 
Since $a \wedge b \leq a, b$, by monotonicity of $\cdot$ it follows that $(a \wedge b) \cdot (a \wedge b) \leq a \cdot b$. 
Hence, it follows that $a \wedge b \leq a \cdot b$.
	 
In order to prove (ii), use that $e$ is the identity of the monoidal structure of $\mathbf{A}$, that is, $a \cdot e = a$. 
By \textbf{A4}, it follows that $a\ \cdot \sim a \leq\ \sim e$. 
By \textbf{A5}, $a\ \cdot \sim a \leq e$.
	 	 
For (iii), from $a \cdot 1 \leq a \cdot 1$ by \textbf{A4} it follows that $a\ \cdot \sim(a \cdot 1) \leq 0$ and then we get our goal by part (i).

Finally, in order to prove (iv), by parts (i) and (ii), and the fact that the monoid reduct of $\mathbf{A}$ is commutative, we have that, $(a\ \cdot \sim b) \wedge (b\ \cdot \sim a) \leq (a\ \cdot \sim b) \cdot (b\ \cdot \sim a) = (a\ \cdot \sim a) \cdot (b\ \cdot \sim b) \leq e \cdot e = e$.	\end{proof}

Recall that in any involutive residuated lattice fusion and the residual are interdefinable; precisely, $x \cdot y :=:\ \sim(x \to\ \sim y)$ and $x \ra y :=:\ \sim(x\ \cdot \sim y)$. 
Hence, the inequality in part (iv) of the previous Proposition may be stated as 

$e \leq (a \to b) \vee (b \to a)$, 

\noindent that is to say, odd Sugihara monoids are prelinear as residuated lattices.
Then, the variety of (bounded) odd Sugihara monoids is generated by its totally ordered members (see Section 3 in \cite{HRT} for details).

In \cite{MRWb}, the authors provide an equational base for the variety generated by the Sugihara three-element chain using the equations 

(21) $e \leq (x \to y) \vee (y \to x)$,

(22) $e \leq [x \to (y \vee \neg y)] \vee (y \wedge \neg y)$.

\noindent In this paper, the variety will be called $\vbtsm$ and its elements will be called (bounded) three-valued Sugihara monoids.
In what follows, we give an alternative equational base for $\vbtsm$.

\begin{defn} \label{tBOSM}
A (bounded) odd Sugihara monoid is said to be \emph{three-valued} if it satisfies
\begin{itemize}
\item[\textbf{A6}] $e \leq x\ \vee \sim(x \cdot 1)$.
\end{itemize}
\end{defn}

In the context of bounded odd Sugihara monoids, {\bf{A6}} is equivalent to equation (22) in \cite{MRWb}. 

\begin{rem}
In any bounded odd Sugihara monoid, equation \textbf{A6} above holds if and only if  the equation {\bf{A6$^\prime$}} given below holds.
\begin{itemize}
\item[\bf{A6$^\prime$}] $(x \cdot (y\ \wedge \sim y)) \wedge (y\ \vee \sim y) \leq e$.
\end{itemize}

\noindent Up to term equivalence, equation {\bf{A6$^\prime$}} is equation (22) in \cite{MRWb}.
\end{rem}

It is clear that the class of bounded three-valued Sugihara monoids forms a variety, which we denote by $\vbtsm$.

\begin{thm}[Theorem 6.2(ii) in \cite{MRWb}]\label{threesugihara}
The variety $\vbtsm$ is generated by the Sugihara chain of three elements $S_3$. 
\end{thm}

As an immediate consequence of the previous theorem we have that any equation which is satisfied in $S_3$ is also satisfied by all the elements of $\vbtsm$.

\subsection{Three-valued \L ukasiewicz algebras}

\vskip.2cm
In 1920, \L ukasiewicz introduced in \cite{I} the notion of three-valued logic as the logic given by the logical matrix $(\textrm{\L}_3, \{1\})$, where $\textrm{\L}_3$ is the three-element Wajsberg hoop. 
It is well known that bounded Wajsberg hoops are term equivalent to MV-algebras. 
Hence, $\textrm{\L}_3$ may alternatively be regarded as the three-element MV-algebra.

In 1940, Moisil introduced the notion of three-valued \L ukasiewicz algebra as an attempt to give an algebraic approach to the three-valued propositional calculus considered by \L ukasiewicz in a different signature. 
Following Monteiro \cite{M}, we can define a three-valued \L ukasiewicz algebra in the following way \cite{I}.

\begin{defn}
A \emph{three-valued \L ukasiewicz algebra} ($3\textrm{\L}$-algebra for short) is an algebra $(A; \wedge, \vee, \sim, \nabla, 0, 1)$ of type (2, 2, 1, 1, 0, 0)  such that
	\begin{itemize}
		\item[]$(A; \wedge, \vee, 0, 1)$ is a bounded distributive lattice,
		\item[]for any $a, b \in A$, it holds that
	%	\begin{enumerate}
			\item[]$\sim \sim a = a$,
			\item[]$\sim (a \wedge b) =\ \sim a \ \vee \sim b$,
			\item[]$\sim a \vee \nabla a = 1$,
			\item[]$a\ \wedge \sim a = \sim a \wedge \nabla a$,
			\item[]$\nabla (a \wedge b) = \nabla a \wedge \nabla b$.
	%	\end{enumerate}	
	\end{itemize}
\end{defn}

It is shown in \cite{M} that any $3\textrm{\L}$-algebra $(A; \wedge, \vee, \sim, \nabla, 0, 1)$ is a Kleene algebra, that is, for every $a, b \in A$, it holds that $a \ \wedge \sim a \leq b \ \vee \sim b$.
\vskip.2cm

Let $(A; \wedge, \vee, \sim, \nabla, 0, 1)$ be a $3\textrm{\L}$-algebra. 
An element $a \in A$ is said to be \emph{central} if $\sim a = a$. 
Since any $3\textrm{\L}$-algebra is a Kleene algebra, the following result is straightforward. 

\begin{lem}
Any $3\textrm{\L}$-algebra has at most one central element.
\end{lem}

We say that an $3\textrm{\L}$-algebra is \emph{centred} if it has a (unique) central element. 
The class of all centred $3\textrm{\L}$-algebras forms a variety in the signature $\{\wedge, \vee, \sim, \nabla, 0, z, 1\}$, which is an expansion with the $0$-ary operation $z$ of that of $3\textrm{\L}$-algebras and satisfies the equality $\sim z = z$. 
Clearly, this variety, which we denote by $\ctl$, is term equivalent to both that of centred three-valued bounded Wajsberg hoops and that of three-valued MV-algebras.

\subsection{Centred three-valued MV-algebras}

\begin{defn}
A \emph{centred three-valued MV-algebra} (MV$_3$-algebra for short) is an algebra $(A; \oplus, \sim, 0, c)$ of type (2, 1, 0, 0)  such that
	\begin{itemize}
		\item[]$(A; \oplus, 0)$ is a commutative monoid and for any $a, b \in A$, it holds that
		%\begin{enumerate}
			\item[]$\sim \sim a = a$,
			\item[]$a \ \oplus \sim 0 =\ \sim 0$,
			\item[]$\sim (\sim a \oplus b) \oplus b = \ \sim (\sim b \oplus a) \oplus a$,
			\item[]$a \oplus (a \oplus a) = a \oplus a$ (see \cite[Equation (MVn1) on p. 433]{AT}),
			\item[]$\sim c = c$.
		%\end{enumerate}	
	\end{itemize}
\end{defn}

\noindent This variety will be denoted $\cmv$.

\subsection{Term equivalence}

We now state the result announced at the beginning of this section.

\begin{thm}\label{MainThm}
The varieties $\hdz$, $\vbtsm$, $\ctl$ and $\cmv$ are term equivalent. 
\end{thm}
\begin{proof}
We will make explicit the term equivalences between the different varieties and leave to the reader the straightforward but tedious labour of completing the details of the proof.
	
Let $(A; \vee, \wedge, \cdot, \sim, 0, e, 1)$ be a bounded three-valued Sugihara monoid. 
For any $x \in A$, define unary operations $\neg x :=\ \sim (x \cdot 1)$ and $Dx :=\ \sim x \cdot 1$, and the $0$-operation $c:= e$. 
Then,  $(A; \wedge, \vee, \neg, D, 0, c, 1)$ is a c3dp-algebra.

Conversely, if $(A; \wedge, \vee, \neg, D, 0, c, 1)$ is a c3dp-algebra and we define a binary operation $\cdot$ by 
\[
x \cdot y := (x \wedge \neg \neg y) \vee (\neg \neg x \wedge y),
\]
a unary operation $\sim$ by
\[
\sim x := Dx \wedge (x \vee \neg x),
\]
and the $0$-ary operation $e:=c$,
then $(A, \vee, \wedge, \cdot, \sim, 0, e, 1)$ is a bounded three-valued Sugihara monoid.

\

The term equivalence between $\hdz$ and $\ctl$ is given, in one direction, by $\sim x := Dx \wedge (x \vee \neg x)$, $\nabla x := \neg \neg x$, and $z:=c$.
In the other direction, it is given by $\neg x :=\ \sim \nabla x$, $D x := \nabla \sim x$, and $c:=z$.
\end{proof}

The term equivalence between $\hdz$ and $\cmv$ is given, in one direction, by $c := c$, $0 := 0$, $1 := \ \sim 0$, $\neg x := \ \sim(x \oplus x)$, $Dx := \ \sim x \ \oplus \sim x$, $x \vee y = [\sim(\sim x \oplus y)] \oplus y$, $x \wedge y = \sim [ \sim x \ \oplus \sim (\sim x \oplus y)]$.  

In the other direction, it is given by $c := c$, $0 := 0$, $\sim x := Dx \wedge (x \vee \neg x)$ and $x \oplus y := (x \vee \neg \neg y) \wedge (\neg \neg x \vee y)$.

\subsection{Categorical equivalence}

Since term equivalent varieties are equivalent categories, the following corollary is an immediate consequence of Theorems \ref{cateq} and \ref{MainThm}.

\begin{cor}
The	varieties $\vbtsm$, $\hdz$, $\ctl$, $\cmv$ and $\bool$ are all equivalent as categories.
\end{cor}

\begin{rem}
Since any algebra in $\hdz$ is functionally complete, it follows that any AGC-algebra (in any good signature) is also functionally complete. 
In particular, the varieties $\vbtsm$ and $\ctl$ are also functionally complete.
\end{rem}

\section{A comparison with recent work} \label{FS}

In \cite{IRp}, the set ${C}(B)$ of all contracts on a Boolean algebra $\bf{B}$ is regarded as a Heyting algebra satisfying Stone equation, that is, equation $\neg x \vee \neg \neg x = 1$, and having a minimum dense element $e = (1,1)$, where, as usual, an element $x$ is \emph{dense} iff $\neg \neg x = 1$. 
If we consider the variety $\aS$ of Stone algebras having a minimum dense element (named \emph{augmented Stone algebras} in \cite{IRp}), 
the assignment $\bf{B} \mapsto \mathcal{C}(B)$ defines a functor $\mathcal{C} : \bool \to \aS$. 
In \cite{IRp} it is shown that $\mathcal{C}$ is part of an adjoint pair $\mathcal{C} \dashv \clos$, where  $\clos(\bf{A})$ is the Stonean subalgebra of an augmented Stone algebra $\bf{A}$ formed by all its complemented elements. 

It is not hard to see that in any pseudo-complemented bounded distributive lattice $(A; \wedge, \vee, \neg, 0, 1)$ having a minimum dense element $\gamma$, the sublattice $[\gamma) = \{a \in A: a \leq \gamma\}$ together with the unary operation $N$ defined by $Na := \neg a \wedge \gamma$, is a Boolean lattice isomorphic to $\clos(A)$. 
Furthermore, in \cite{IRp}, the author could have taken the functor $(\ )^{+} : \aS \to \bool$ defined by the assignment $A \mapsto [\gamma)$ instead of $\clos$. 
Clearly, we also have that $\mathcal{C} \dashv (\ )^{+}$.
 
Due to the functional completeness of 3cdp-algebras, it follows that $\bf{C}(B)$ has the underlying structure of a Heyting algebra, which is Stonean and has a minimum dense element. 
As a consequence, we have a forgetful functor $\bf{U} : \hdz \to \aS$ making the following diagram commute.

	\[
	\xymatrix{
		&\bool  \ \ \ \ar@<1ex>[rr]^{\mathcal{C}}_{\bot} \ar@<1ex>[ddr]^{\mathbf{C}} &   & \ \ \ \aS \ar@<1ex>[ll]^{(\ )^{+}}       \\
		&                                             &   &   \\
        &    &  \hdz \ar@<1ex>[uul]^{(\ )^{-}}  \ar[uur]_{U}&   \\
	}
	\] 	

\noindent Since $\bf{C}$ and $(\ )^{-}$ witness an equivalence, it follows that $(\ )^{-} \dashv \mathbf{C}$ and, in consequence, $ \bf{U} = \mathcal{C} \circ \bf{(\ )^{-}} \dashv \mathbf{\bf{C}} \circ (\ )^{+} \cong \mathbf{\bf{C}} \circ \clos$. \\

\noindent {\bf{Acknowledgement}}

\noindent The first author was supported by Consejo Nacional de Investigaciones Cient\'ificas y T\'ecnicas (PIP 11220200100912CO, CONICET-Argentina), Universidad Nacional de La Plata (PID-X921) and Agencia Nacional de Promoci\'on Cient\'ifica y Tecnol\'ogica (PICT2019-2019-00882, ANPCyT-Argentina).


\begin{thebibliography}{8}

\bibitem{CE} Castiglioni, J.L. and Ertola-Biraben, R.C. (2024). 
Modalities combining two negations. 
\emph{Journal of Logic and Computation}, exae004, 
%https://doi.org/10.1093/logcom/exae004 (2024).

\bibitem{ABCG} Aguzzoli, S., Davide Ciucci, D., Boffa S. and Gerla,  B. (2018).
Finite IUML-algebras, Finite Forests and Orthopairs.
\emph{Fundamenta Informaticae} {\bf{163}}: 139--163.  

\bibitem{BCNP} Benveniste, A., Caillaud, B., Nickovic, D., Passerone, R., Raclet, J.-
B., Reinkemeier, P., Sangiovanni-Vincentelli, A., Damm, W., Henzinger, T. A., and
Larsen, K. G. (2018). 
\emph{Contracts for system design}. Foundations and Trends in Electronic Design Automation, {\bf{12}}(2-3): 124--400.

\bibitem{DC} Ciucci, D. (2011)
Orthopairs: A Simple and Widely Used Way to Model Uncertainty.
\emph{Fundamenta Informaticae} {\bf{108}}: 287--304
DOI 10.3233/FI-2011-424

 \bibitem{FG} Fussner, W. and Galatos, N. (2019).
 Categories of models of R-mingle. 
\emph{Ann. Pure Appl. Logic} {\bf{170}}:1188--1242.

\bibitem{GR12}
Galatos, N. and  Raftery, J.G. (2012).
A category equivalence for odd Sugihara monoids and its applications. 
{\em{Journal of Pure and Applied Algebra}} {\bf 216}:  2177--2192.

\bibitem{GR15} Galatos, N. and Raftery, J.G. (2015).
Idempotent residuated structures: some category equivalences and their applications. 
\emph{Trans. Amer. Math. Soc.} {\bf{367}}: 3189-3223.

\bibitem{HRT}
Hart, B. J., Rafter, L. and Tsinakis, C. (2000).
Commutative Residuated Lattices. 
On-line notes: 
https://cdn.vanderbilt.edu/vu-my/wp-content/uploads/sites/1171/2013/11/14113136/reslat.pdf

\bibitem{CH}
C. A. R. Hoare, C.A.R. (1969).
An Axiomatic Basis for Computer Programming
Communications of the ACM Volume 12 / Number 10 / October, 1969.

\bibitem{IR}
Incer Romeo, I.X. (2022) 
\textit{The Algebra of Contracts}. Ph.D. Thesis, UC Berkeley. 
https://escholarship.org/uc/item/1ts239xv.

\bibitem{IBV}
Inigo Incer, Albert Benveniste, and Alberto Sangiovanni-Vincentelli. (2023).
Some Algebraic Aspects of Assume-Guarantee Reasoning (Preprint).
{https://arxiv.org/abs/2309.08875}

\bibitem{IRp}
Incer, I. (2024).
An Adjunction Between Boolean Algebras and a Subcategory of Stone Algebras. 
\emph{Theory and Applications of Categories} {\bf{41}}(57): 2041-2057.

\bibitem{I}
Iturrioz, L. (1977). 
An axiom system for three-valued {\L}ukasiewicz propositional calculus. 
{\em{Notre Dame Journal of Formal Logic}} {\bf XVIII}(4): 616--620.
	
\bibitem{TK} 
Katri\v{n}\'ak, T. (1973).
The structure of distributive double \emph{p}-algebras. 
\emph{Algebra Universalis} {\bf{3}}: 238--246.

\bibitem{AM} Malcev, A. (1954).
On the general theory of algebraic systems (Russian). 
\emph{Mat. Sb. N.S.} {\bf{35}}(77): 3-20. 

\bibitem{BM} Meyer, B. (1992).
Applying ``Design by Contract".
\emph{Computer} {\bf{25}}(10): 40--51. 
DOI: 10.1109/2.161279

\bibitem{MRW}
Moraschini, T., Raftery, J.G. and Wannenburg, J.J. (2020).
Varieties of De Morgan Monoids: Covers and Atoms. 
{\em{The Review of Symbolic Logic}} {\bf{13}}(2): 338--374. 
doi.org/10.1017/S1755020318000448

\bibitem{MRWb}
Moraschini, T., Raftery, J.G. and Wannenburg, J.J. (2019).
Varieties of De Morgan Monoids: Minimality and irreducible algebras. 
\emph{Journal of Pure and Applied Algebra} {\bf 223}(7): 2780--2803.
https://doi.org/10.1016/j.jpaa.2018.09.015

\bibitem{M}	
Monteiro, A. (1963).
Sur la definition des algebres de L ukasiewicz trivalentes.
{\em{Bulletin Mathematique de la Societe Scientifique Mathematique Physique R. P. Roumanie}}, {\bf{7}} (55): 3--12.

\bibitem{TS} 
Skolem, T. (1919). 
\emph{Untersuchungen \"uber die Axiome des Klassenkalk\"uls und \"uber
Produktations- und Summations-Probleme, welche gewisse Klassen von Aussagen betreffen}, Vol. 3 of \emph{Videnskapsselskapet Skrifter, I. Matematisk-naturvidenskabelig Klasse}, In
Kommission bei Jacob Dybwad, Kristiania.

\bibitem{AT}
Torrens, A. (1994).
Cyclic Elements in MV-Algebras and Post Algebras.
\emph{Math. Log. Quart.} {\bf{40}}: 431-444.

\bibitem{V1972} 
Varlet, J. (1972). 
A regular variety of type $<2, 2, 1, 1, 0, 0>$.
\emph{Algebra Universalis} {\bf{2}}: 218--223.
	
\end{thebibliography}
\end{document}